\documentclass[letterpaper]{article} 
\usepackage{aaai23}  

\nocopyright

\usepackage{times}  
\usepackage{helvet}  
\usepackage{courier}  
\usepackage[hyphens]{url}  
\usepackage{graphicx} 
\usepackage{natbib}  
\usepackage{caption} 
\usepackage{algorithm}
\usepackage{algpseudocode}

\usepackage[dvipsnames]{xcolor}

\usepackage{booktabs}

\usepackage{xspace}
\newcommand{\scom}{\texttt{SCOM}\xspace}
\newcommand{\scoms}{\texttt{SCOMs}\xspace}
\newcommand{\cbm}{\texttt{CBM}\xspace}
\newcommand{\cbms}{\texttt{CBMs}\xspace}
\newcommand{\cubsel}{\texttt{CUB-Sel}\xspace}

\usepackage{amsthm}
\usepackage{amsmath}
\usepackage{amssymb}
\DeclareMathOperator*{\argmax}{arg\,max}
\DeclareMathOperator*{\argmin}{arg\,min}

\newtheorem{theorem}{Theorem}

\newcommand{\sut}{\quad \text{subject to} \quad}
\newcommand{\argmaxcz}{\argmax_{\mathcal{C} \subset \mathcal{Z}}}
\newcommand{\argmincz}{\argmin_{\mathcal{C} \subset \mathcal{Z}}}

\newcommand{\argmaxciz}{\argmax_{c_i \in \mathcal{Z}}}
\newcommand{\argminciz}{\argmin_{c_i \in \mathcal{Z}}}

\newcommand{\argmaxcic}{\argmax_{c_i \in \mathcal{C}}}
\newcommand{\argmincic}{\argmin_{c_i \in \mathcal{C}}}

\usepackage{newfloat}
\usepackage{listings}
\DeclareCaptionStyle{ruled}{labelfont=normalfont,labelsep=colon,strut=off} 
\floatstyle{ruled}
\newfloat{listing}{tb}{lst}{}
\floatname{listing}{Listing}
\title{Selective Concept Models:\\Permitting Stakeholder Customisation at Test-Time}
\author {
    Matthew Barker,\textsuperscript{\rm 1}
    Katherine M. Collins, \textsuperscript{\rm 1}
    Krishnamurthy (Dj) Dvijotham, \textsuperscript{\rm 2}
    Adrian Weller, \textsuperscript{\rm 1 \rm 3}
    Umang Bhatt \textsuperscript{\rm 1 \rm 3}
}
\affiliations {
    \textsuperscript{\rm 1} University of Cambridge\\
    \textsuperscript{\rm 2} Google Research (Brain Team)\\
    \textsuperscript{\rm 3} Alan Turing Institute
}

\usepackage{bibentry}

\begin{document}

\maketitle

\begin{abstract}
Concept-based models perform prediction using a set of concepts that are interpretable to stakeholders. However, such models often involve a fixed, large number of concepts, which may place a substantial cognitive load on stakeholders. We propose Selective COncept Models (\scoms) which make predictions using only a subset of concepts and can be customised by stakeholders at test-time according to their preferences. We show that \scoms only require a fraction of the total concepts to achieve optimal accuracy on multiple real-world datasets. Further, we collect and release a new dataset, \cubsel, consisting of human concept set selections for 900 bird images from the popular CUB dataset. Using \cubsel, we show that humans have unique individual preferences for the choice of concepts they prefer to reason about, and struggle to identify the most theoretically informative concepts. The customisation and concept selection provided by \scom improves the efficiency of interpretation and intervention for stakeholders.
\end{abstract}

\section{Introduction}
Humans can reason about a limited number of concepts at once when making decisions \citep{miller1956magical,luck1997capacity,cowan2001magical}. While concept-based methods such as Concept Bottleneck Models (\cbms) \citep{koh2020concept} have been proposed to support human interpretability and intervenability in machine learning (ML) systems, such models typically involve dozens of concepts, well beyond the number of concepts stakeholders can process at any given time~\citep{tenenbaum1998bayesian,ramaswamy2022overlooked}.

To reduce the cognitive load of reasoning about many concepts, we propose Selective COncept Models (\scoms). \scoms provide a streamlined extension of \cbms by selecting the concepts that are most pertinent to any given task from a larger set of available concepts. This enables a stakeholder to reason with a reduced set of concepts without compromising task accuracy. Unlike \cbms which require a fixed concept set, \scoms make predictions using an arbitrary concept subset which can be customised at inference-time \textit{ without retraining}. For example, one might wish to prohibit consideration of sensitive attributes such as biological sex and subjective attractiveness during prediction. For \scoms, withdrawing certain concepts is trivial, whereas conventional \cbms make such exclusion difficult.

\scoms enable flexible customisation according to a stakeholder's preferences for the number of concepts to use and their personal trade-off between cognitive load and predictive accuracy \citep{ramaswamy2022overlooked}. On the task of bird species recognition, \scoms require only 6 out of 28 concepts to achieve optimal prediction accuracy. Smaller concept sets decrease the human cost of interventions, and increase the impact of each intervention. Thus this work is complementary to research aimed at designing better intervention policies over a given set of concepts, for example CooP \citep{chauhan2022interactive}. Since \scoms place no restrictions on the exact output network architecture, they provide a simple extension to existing models used by practitioners.

\begin{figure*}[t]
    \centering
    \includegraphics[width=1.0\linewidth]{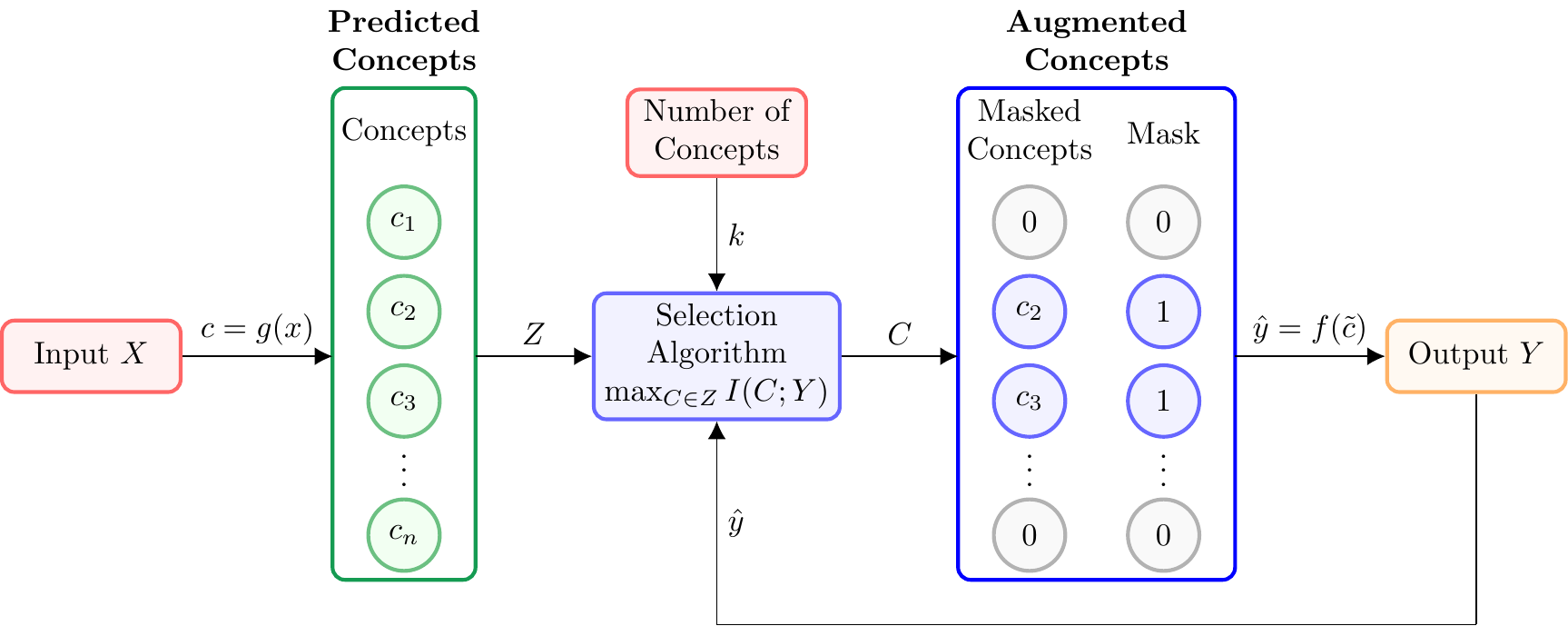}
    \caption{In \cbms, human-interpretable concepts are predicted from the input $X$ and used to infer the output $Y$. \scoms extend \cbms by selecting the most relevant concepts which maximise the mutual information between the selected concepts ($C$) and the output ($Y$). The output model learns to make predictions using an augmented concept vector, which contains the mask used to select concepts. This allows the number of concepts ($k$) and the specific concepts selected to be customised at inference time by the stakeholder without retraining the model.}
    \label{fig:figure1}
\end{figure*}

\section{Contributions}
\begin{itemize}
    \item We show how feature selection can be applied to the problem of concept selection, using mutual information (MI) maximisation to provide ``useful'' selections, even in the presence of duplicated concepts.
    \item We propose \scoms; models which make predictions using an arbitrary subset of concepts, allowing customisation by stakeholders without retraining.
    
    \item We show that \scoms provide comparable accuracy on real-world datasets CUB and CelebA, while using only a fraction of the available concepts. We confirm that \scom selection algorithms outperform existing techniques for feature selection.

    \item We design a human subject experiment to validate the use of \scoms, and gather 900 sets of concept selections, which we release as \cubsel. We identify that people have varied preferences for the number of concepts they select, and favour concepts which are less theoretically informative for the task. This motivates the use of \scoms to adaptively select lightweight concept sets.
    
\end{itemize}

\section{Definitions}
We follow \citep{koh2020concept} in considering concepts as high-level, human-interpretable -- and human-specifiable -- values. We define $\mathcal{Z}$ as the complete set of available concepts, $\mathcal{C} \subset \mathcal{Z}$ as the chosen set of concepts of size $k$, and $\mathcal{Y}$ as the target predictions. \scoms aim to identify the set $\mathcal{C}$ for a given $k$ such that the performance for a mapping $\mathcal{C} \rightarrow \mathcal{Y}$ is comparable to the mapping $\mathcal{Z} \rightarrow \mathcal{Y}$. We largely focus on the setting where $\mathcal{Z}$ is obtained from a Sequential \cbm \citep{koh2020concept}, which predicts the concepts based on inputs $\mathcal{X}$ (such as image pixels). However, \scoms be applied to the inputs $\mathcal{X}$ directly in the case where these inputs are also human concepts. It is assumed that $\mathcal{Z}$ is given for the \scom. We use capitalised letters (e.g. $X$) to denote random variables and lowercase letters to denote instances of a random variable (e.g. $x$).

\section{Comparison to Prior Work}
\label{sec:prior-work}
\subsubsection{Concept-Based Models}
There has been a proliferation of concept-based models which use concepts as intermediaries for output predictions. There are several flavours of these models: those that rely on datasets annotated with concept attributes \citep{koh2020concept,marconato2022glancenets,zarlenga2022concept}, those that learn concepts \citep{alvarez2018towards}, and those that introduce concept features post-training \citep{yuksekgonul2022post}. Most models emphasise the motivation of using human-understandable concepts, yet unlike \scoms, typically ignore the overall size of the concept set. 



\subsubsection{Model Intervenability}
\citep{chauhan2022interactive,shin2022closer} give ways of making interventions in \cbms more effective than random: their work can be used \emph{in conjunction} with \scom. \scom reduces the size of the concept set used for prediction, increasing the impact of each intervention, while still allowing intervention procedures to select the best concept to intervene on. While in practice, stakeholders may prefer to provide soft labels if they are unsure, \citep{collins2023human} show that existing models perform poorly with uncertain interventions.

\subsubsection{Model Interpretability}
There are a wide range of techniques which attempt to improve the interpretability of concept-based models \citep{kim2018interpretability,zhou2018interpretable,ghorbani2019towards,yeh2020completeness,crabbe2022concept,abid2022meaningfully,bai2022concept}. These techniques often provide metrics highlighting the importance of relevant concepts to the output prediction. However, \citep{kim2022help} show with human feedback experiments that concept-based models using numbers for interpretability can be overwhelming and lose relevance to the actual image. \scoms are compatible with these existing interpretability methods, while reducing the number of concepts a human has to consider.



\subsubsection{Feature Engineering and Selection}
Feature engineering, the process of applying a set of hand-crafted transforms to raw input data, is commonly used to increase model accuracy and robustness of ML models \citep{zheng2018feature}, with significant success across domains including text modelling \citep{lewis1992representation} and computer vision \citep{nixon2019feature}. By considering concepts as human-interpretable features, concept selection in \scoms is an application of feature engineering. Information-theoretic methods, commonly used to maximise the statistical dependency between features and the output, are well-suited to the problem of selecting concept sets in \scoms. \citep{fleuret2004fast} have empirical success in approximating the best global feature set using conditional MI.

 Instead of selecting one feature set for the whole dataset, instance-level techniques can be used to identify custom feature sets for each datapoint. For example, Learning to Explain (L2X) \citep{chen2018learning} trains a neural network to predict a variational approximation to MI, and uses a sampler to select the feature set of a given size which maximises this quantity. Unlike \scoms, existing techniques either provide no control over the feature set size (e.g. INVASE \citep{yoon2019invase}), or require the size to be fixed during training (e.g. L2X). To the best of our knowledge, \scoms are the first models which allow customisation for instance-level selection \emph{without retraining}.


\section{Approach}
\label{sec:approach}
In this section, we provide a high-level overview for each of the four key stages in \scom: 
\begin{enumerate}
    \item \textbf{Concept model (training).} Train a model $x \rightarrow c$ which predicts the concepts ($c$) present in each input ($x$):
    \begin{equation}
        c = g(x)
    \end{equation}
    For some datasets, the concepts may be the input $x$ itself.
    \item \textbf{Output model (training).} Train a neural network which makes predictions using any arbitrarily chosen concept set $\mathcal{C} \in \mathcal{Z}$. The original concept vector is augmented to give $\tilde{c}$ according to the chosen subset, which is then used for prediction:
    \begin{equation}
        \hat{y} = f(\tilde{c})
    \end{equation}
    \item \textbf{Concept selection (test).} Select the best concept set $\mathcal{C} \subset \mathcal{Z}$ for each size $k$, using a greedy algorithm which maximises MI between the concepts and the outputs:
    \begin{equation}
        C_s = \argmaxcz I(Y;C) \sut |C| = k
    \end{equation}
    The MI is approximated using the trained network in Step 2. 
    \item \textbf{Inference (test):} A value of $k$ is chosen and inference is performed using the selected concept set of size $k$:
    \begin{equation}
        y^* = f(\tilde{c}_s)
    \end{equation}
    Stakeholders may customise the concept set depending on their individual preferences.
\end{enumerate}

\scom addresses two main issues: \emph{how to select a concept set for predictions}, and \emph{how to perform inference without retraining the model.} These two problems are interlinked; MI estimation in Step 3 necessitates inference with a concept set to avoid retraining, while test-time prediction requires concept selection to find an efficient concept set for a given $k$. Since concept selection is at the instance level, it occurs at test time. Further details on each stage are provided next.

\subsection{Training the Concept Model}
\label{sec:concept-model}
The concept model predicts the concepts present in each input. \scoms impose no restrictions on the details of this model, and only require concept predictions for training further stages and inference. As a result, the concept model may take on a variety of forms:
\begin{itemize}
    \item \textbf{Sequential Bottleneck.} Using human-annotated datapoints (input $x$, concepts $c$), a sequential bottleneck \citep{koh2020concept} trains a model $x \rightarrow c$ which predicts the concepts present in each input using a neural network. This is a direct application of multi-label classification in the supervised learning setting, which is well-studied in the literature \citep{zhang2013review}.
    \item \textbf{Data processing.} In cases when it is not possible to obtain annotated concepts for each datapoint, data-processing and augmentation techniques may be used to predict concepts. This is common in the computer vision domain \citep{perez2017effectiveness}, although it is widely used across other domains as well \citep{kotsiantis2006data}.
    \item \textbf{Direct Inputs.} For certain datasets, such as tabular data, the inputs may already be human-interpretable and thus considered as concepts directly. In these cases, the concept model for \scom will simply be the identity function, i.e. $\hat{c} = x$.
\end{itemize}
The further stages below assume the concepts have already been obtained from each datapoint.

\subsection{Training the Output Model}
\label{sec:output-model}
The output model predicts the output label using \emph{any concept set} $\mathcal{C} \subset \mathcal{Z}$. In comparison to the concept model, there is a lack of prior work on models which can perform inference using arbitrary predictions. \citeauthor{yoon2019invase} attempt to address this challenge by ``masking out'' any unused inputs by replacing their values with zeros. However, this approach prevents the network from distinguishing between ``masked zeros'' and input zeros.

\subsubsection{Concept Augmentation}
Instead, \scoms augment the concept vector with the binary mask vector used to mask unused concepts, shown in Figure \ref{fig:figure1}. Thus the network has enough information to identify which zeros are masked, and which are genuine inputs. The mask is applied to the concept vector $c$ to give the augmented vector:
\begin{equation}
    \tilde{c} = \begin{bmatrix} c \odot m \\ m \end{bmatrix}
\end{equation}
Where $\odot$ denotes the elementwise Hadamard product. For example, consider the case with 3 concepts, and concepts 1 and 3 are selected. In this case, the concept vector, mask, and augmented concept vector are given below:
\begin{equation}
    c = \begin{bmatrix}
        \textcolor{red}{c_1} \\
        \textcolor{red}{c_2} \\
        \textcolor{red}{c_3}
    \end{bmatrix}, \qquad
    m = \begin{bmatrix}
        \textcolor{blue}{1} \\
        \textcolor{blue}{0} \\
        \textcolor{blue}{1}
    \end{bmatrix}, \qquad
    \tilde{c} = \begin{bmatrix}
        \textcolor{red}{c_1} \\
        \textcolor{red}{0} \\
        \textcolor{red}{c_3} \\
        \textcolor{blue}{1} \\
        \textcolor{blue}{0} \\
        \textcolor{blue}{1}
    \end{bmatrix}
\end{equation}
For multi-dimensional concepts, each value in the concept is multiplied by the corresponding mask value. Thus, \scoms can select concepts of any dimension, unlike many existing methods. Once the model has constructed $\tilde{c}$ from the chosen subset, the output model makes predictions using $\tilde{c}$ as input:
\begin{equation}
    \hat{y} = f(\tilde{c})
\end{equation}
The function $f$ is complex and can be approximated using a neural network, the training of which is described below.

\subsubsection{Training the Model}
The output model learns a mapping from all possible subsets of $\mathcal{Z}$ to the output $\mathcal{Y}$, i.e. $\mathcal{P}(\mathcal{Z}) \rightarrow \mathcal{Y}$ where $\mathcal{P}$ denotes the power set. To achieve this, a new mask $m$ is sampled for each batch, which then augments the concept vector. The predictions are used to calculate the chosen loss (e.g. cross-entropy), and the model parameters $\theta$ are updated using Stochastic Gradient Descent (SGD). The training procedure is detailed in Algorithm \ref{alg:network} below.

\begin{algorithm}
\caption{Training output model}\label{alg:network}
\begin{algorithmic}[1]
\For{Batch $B$ in $[B_1,...,B_N]$}
\State $m \sim p(m)$ \Comment{Sample a mask $m$}
\State $c \gets c \odot m$ \Comment{Hadamard Product $\odot$}
\State $\tilde{c} \gets [c, m]$ \Comment{Append mask to input vector}
\State $\hat{y} = f(\tilde{c})$ \Comment{Predictive function $f$}
\State $\theta \gets SGD(\hat{y}, y, \theta)$ \Comment{Update $\theta$ using SGD}
\EndFor
\end{algorithmic}
\end{algorithm}

It is illustrative to compare the \scom output model which learns $\mathcal{P}(\mathcal{Z}) \rightarrow \mathcal{Y}$ to a \cbm which learns $\mathcal{Z} \rightarrow \mathcal{Y}$. Since \scom requires a more complex mapping than \cbm it requires greater model complexity in the form of larger/additional hidden layers. In addition, the concept augmentation doubles the size of the input layer in comparison to a \cbm.

The probability distribution $p$ used to sample each mask needs to be determined when training the model. To minimise the expected loss, the distribution of the mask used during training should match the true mask distribution. Consider rewriting $p(m)$ using the dependence on $k$:
\begin{align}
    p(m) &= \sum_{k=1}^n p(m, k) \\
    &= \sum_{k=1}^n p(m | k) p(k)
\end{align}
During training, $p(m | k)$ and $p(k)$ are unknown and for \scoms we assume both distributions are uniform. The two step sampling procedure is as follows:
\begin{enumerate}
    \item Sample $k \sim U(1, k)$
    \item Randomly set $k$ values of mask $m$ to 1, with equal probability.
\end{enumerate}
Importantly, this ensures the output model learns to make predictions using all possible concept set sizes equally. In cases where the distribution of $k$ is known prior to training, using this distribution when sampling $k$ may achieve faster convergence and a lower expected loss. However, throughout this work $k$ is assumed to be uniformly distributed.

\subsection{Concept Selection}
\label{sec:concept-selection}
We define \emph{concept selection} as the method of choosing the concept set in a way that balances the number of concepts used for a given task and downstream task performance. \citep{ghorbani2019towards} propose that chosen concepts should satisfy the three desiderata of meaningfulness, coherency and importance. In the case of \scoms, concepts are provided as dataset annotations which are chosen by humans to be meaningful and coherent. However, the importance of a concept determines how necessary it is to make accurate predictions and is task dependent.

We use MI between the concepts and the labels, $I(Y;C)$, as a measure of the collective prediction accuracy of the concept set. Using an MI measure takes into account the correlations between concepts, rather than ranking every concept in isolation. For example, while concepts such as ``belly colour" and ``wing colour" might both be important individually for predicting the bird species, they are likely to be strongly correlated. Therefore, only one concept may be required to make accurate predictions. Precisely, the task of selecting $k$ concepts can be framed as one of maximising MI subject to a cardinality constraint \citep{fleuret2004fast}:
\begin{equation}
    \label{eqn:cmim}
    \argmaxcz I(Y;C) \sut |C| = k
\end{equation}

The naive approach of trying every subset of $\mathcal{Z}$ with size $k$ suffers from combinatorial explosion and is unfeasible. However, we can achieve a \emph{good} solution in polynomial time using greedy algorithms \citep{nemhauser1978analysis}. There are two common greedy algorithms: forward selection (FS) and backward elimination (BE). FS starts with an empty set $\mathcal{C} = \emptyset$, and then sequentially adds individual concepts $c \in \mathcal{Z} $, maximizing $I(Y; C \cup c_i)$ at every stage. In comparison, BE starts with the complete set $\mathcal{C} = \mathcal{Z}$, and removes individual concepts $c \in \mathcal{C}$, maximizing $I(Y; C \setminus c_i)$ at every stage. The two procedures are detailed in Algorithms \ref{alg:forward-selection} and \ref{alg:backward-elimination}.

\begin{algorithm}
\caption{Forward selection procedure}
\label{alg:forward-selection}
\begin{algorithmic}[1]
\State $Z \gets \{c_1, c_2,...\}$ \Comment{Complete set of available concepts}
\State $C \gets \emptyset$ \Comment{Start with empty set}
\While{$|C| \leq k$}
\State $c_{\text{next}} \gets \argmaxciz I(Y;C \cup c_i)$ \Comment{Maximise MI}
\State $C \gets C \cup c_{\text{next}}$
\EndWhile
\end{algorithmic}
\end{algorithm}

\begin{algorithm}
\caption{Backward elimination procedure}
\label{alg:backward-elimination}
\begin{algorithmic}[1]
\State $Z \gets \{c_1, c_2,...\}$ \Comment{Complete set of available concepts}
\State $C \gets Z$ \Comment{Start with complete set}
\While{$|C| \geq k$}
\State $c_{\text{next}} \gets \argmaxcic I(Y;C \setminus c_i)$ \Comment{Maximise MI}
\State $C \gets C \setminus c_{\text{next}}$
\EndWhile
\end{algorithmic}
\end{algorithm}

Due to the incremental nature of greedy selection, every concept set of size $ \leq k$ (for FS) and $ \geq k$ (for BE) is also obtained when finding the concept size of size $k$. This allows concept sets of all sizes to be calculated by running the selection procedure once for $k = |Z|$ (FS) and $k=0$ (BE). Since \scoms place \emph{no constraint on the dimensionality of concepts}, calculating $I(Y; C)$ requires an approximation when concepts are continuous. Existing methods for estimating MI between high-dimensional continuous variables such as the kernel density estimator \cite{kolchinsky2017estimating} or k-nearest neighbours method \cite{kraskov2004estimating} are sensitive to exact parameters and are not well suited to distributions of concept logits, which tend to be bimodal for binary concepts. To avoid this issue, rather than maximising MI directly, in Theorem \ref{theorem:output-entropy} we show that minimising the output entropy $H(\hat{Y}|C)$ instead gives the same optimal concept set. Unlike MI, output entropy is simple to calculate from the model predictions:
\begin{equation}
    H(\hat{Y}|C) = -\sum_{y \in \mathcal{Y}} p(\hat{y}|c) \log p(\hat{y}|c) 
\end{equation}
Using entropy as a proxy for MI, the maximisation at each stage in FS (Algorithm \ref{alg:forward-selection}) can be solved using:
\begin{equation}
    c_{\text{next}} = \argmaxciz I(Y;C \cup c_i) = \argminciz H(Y | C \cup c_i)
\end{equation}
Similarly, for BE (Algorithm \ref{alg:backward-elimination}) the optimisation at each stage is given by:
\begin{equation}
    c_{\text{next}} = \argmaxcic I(Y;C \setminus c_i) = \argmincic H(Y | C \setminus c_i)
\end{equation}

\subsection{Inference}
\label{sec:inference}
A core property of \scoms is that predictions can be made from arbitrarily chosen concept sets -- \textit{without retraining the model}. Once the output model has been trained (Step 2) and the concept set has been selected (Step 3), inference is simple and follows a three stage procedure:
\begin{enumerate}
    \item The concepts $c$ are predicted from the test datapoint $x^*$:
    \begin{equation}
        c = g(x^*)
    \end{equation}
    \item A subset of concepts, $c_s$, is chosen according to the selection algorithm, and optionally customised by a stakeholder.
    \item The concept subset is augmented to give $\tilde{c}_s$ and used to infer the output:
    \begin{equation}
        y^* = f(\tilde{c_s})
    \end{equation}
\end{enumerate}
\section{Computational Experiments}
\label{sec:real-datasets}
To validate the \scom selection algorithms compared to existing methods, we assess the performance of \scoms on datasets for concept-based models popular in the literature: the CUB dataset for bird species recognition~\citep{welinder2010caltech} and the CelebFaces Attributes Dataset (CelebA) \citep{liu2015faceattributes} for facial identity recognition. For both datasets, a multi-layer perceptron with two hidden layers, each containing 100 neurons, was trained for the output model. Using all the concepts, a prediction accuracy of 75.3\% and 63.8\% was achieved for CUB and CelebA respectively. Further details for each dataset are provided below.

\subsubsection{CUB}
\label{sec:CUB}
We follow \citep{koh2020concept} in using only 112 of the 312 original CUB binary attributes (e.g. ``has blue wings", ``has a pointed beak"). These attributes are further categorised into the 28 multi-valued concepts (e.g. ``wing colour'') used in the original crowdsourced annotation by \citep{welinder2010caltech}. For the purpose of intervention and interpretation, one multi-valued concept more accurately represents a human ``concept" than the binary attributes. Matching \citep{koh2020concept}, the concept labels are fixed at the class-level, so that each bird from the same species is labelled with the same concept annotations. This property, along with large number of concept groups, makes CUB particularly well-suited for \scoms. Further, the same predicted concept logits for the Sequential bottleneck were used when training and evaluating \scom performance, which achieve a concept accuracy of 96.6\%.

Since the CUB concepts are multi-valued, methods which select individual values (e.g. L2X) aren't suitable. The prediction accuracy of randomly selected concepts is used as a baseline. Errors are reported over 3 random seeds which determine the concept logit predictions, with the train/test split given by \citep{wah2011caltech}.

\subsubsection{CelebA}
CelebA \citep{liu2015faceattributes} is a dataset with 40 concept annotations for 200K celebrity faces from 10K people, which has been used to evaluate concept models \citep{ghorbani2019towards,zarlenga2022concept}. For computational efficiency, we subsample CelebA by using random subsets containing 20 identities for the prediction task. Unlike the processed version of CUB, attributes for CelebA are determined at the image level, meaning different images of the same person may have different concept annotations. Although the annotations are noisy, attributes such as ``wearing a hat" cannot be sensibly fixed at the class level. Thus CelebA is used to investigate the suitability of \scoms for instance-level concept annotations. A ResNet-50 \citep{he2016deep} was trained on all 10k identities to predict the attributes for each face, with an accuracy of $0.916$.

Unlike CUB, the concepts are single-valued attributes, and thus L2X was used as a benchmark. Since L2X cannot make predictions directly, concepts are selected by L2X and then used as inputs to the trained \scom output model. Errors are reported over 5 random seeds which determine the samples taken from the dataset.

\subsection{Results}
\subsubsection{Prediction Accuracy}
\begin{figure*}[t!]
    \centering
    \includegraphics[width=0.49\linewidth]{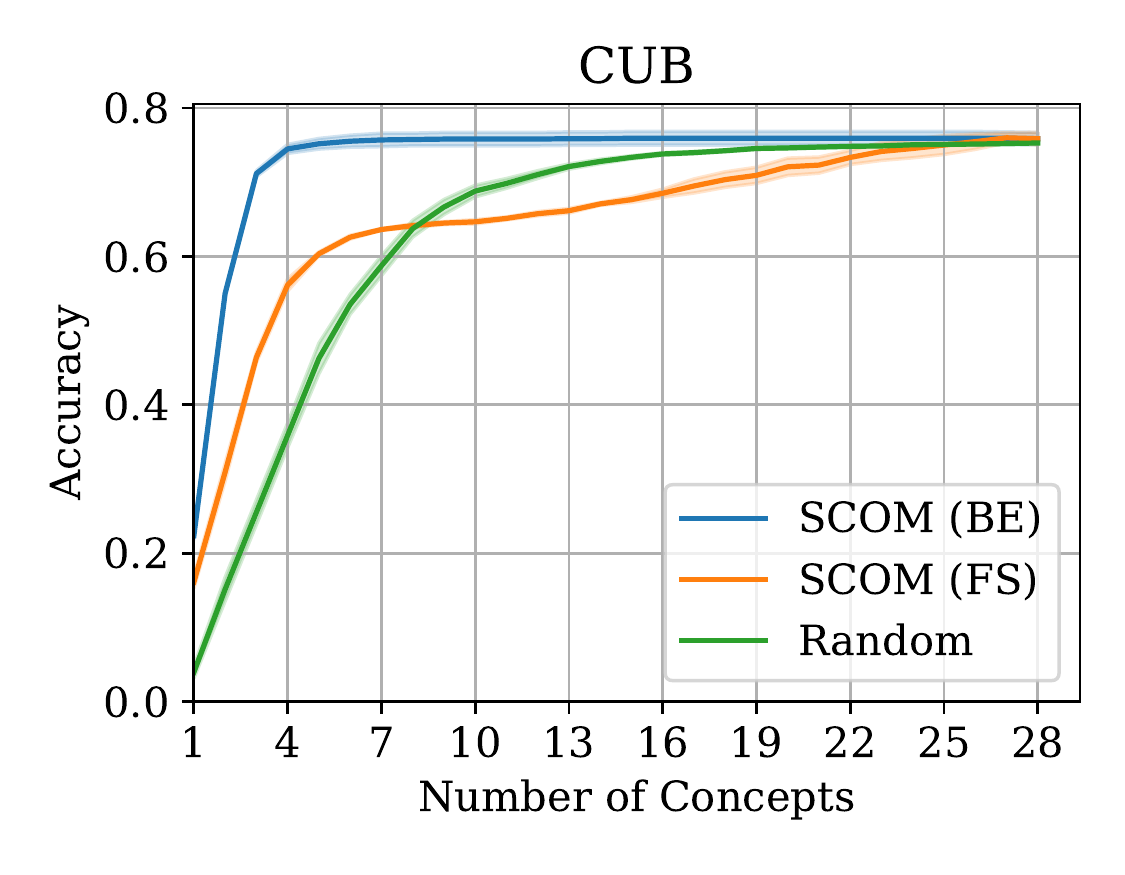}
    \includegraphics[width=0.49\linewidth]{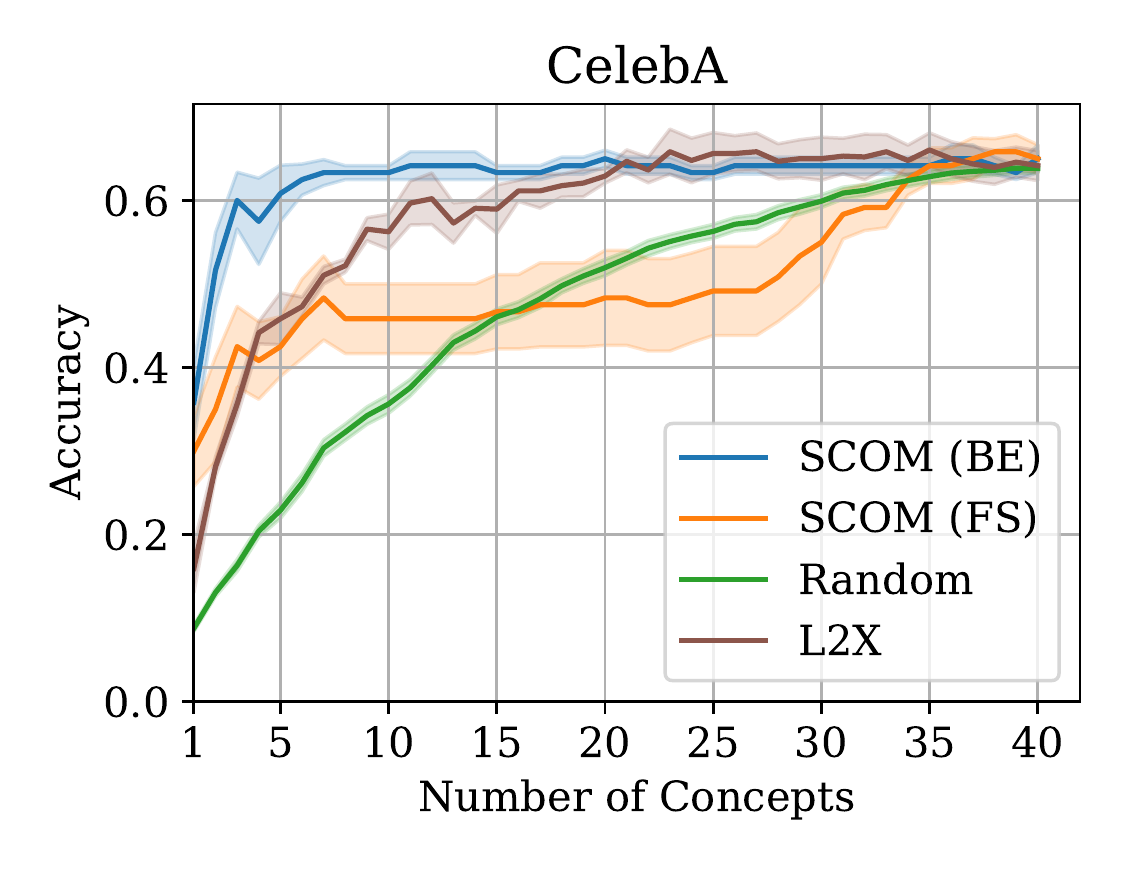}
    \caption{Prediction accuracy of \scoms on the CUB and CelebA datasets as a function of $k$, compared to L2X and random baselines. On CUB (left), BE \scom performs significantly better than FS \scom at the instance level. On CelebA (right), BE \scom performs much better than FS \scom, and better than L2X for low numbers of concepts. For all plots, errors are $\pm 1\sigma$, calculated over random seeds which affects the concepts randomly selected, the concept model chosen for CUB, and the dataset sample for CelebA.}
    \label{fig:prediction-accuracy}
\end{figure*}
\begin{figure*}[t!]
    \centering
    \includegraphics[width=0.49\linewidth]{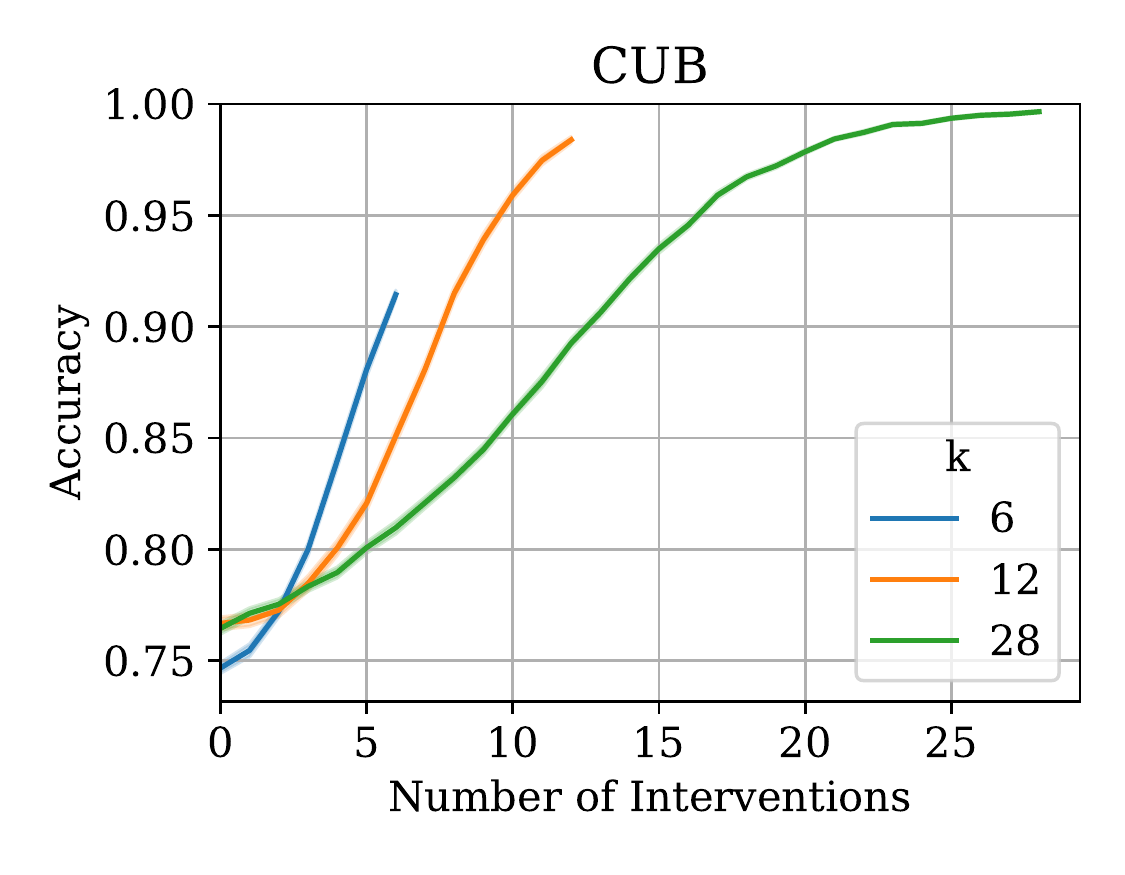}
    \includegraphics[width=0.49\linewidth]{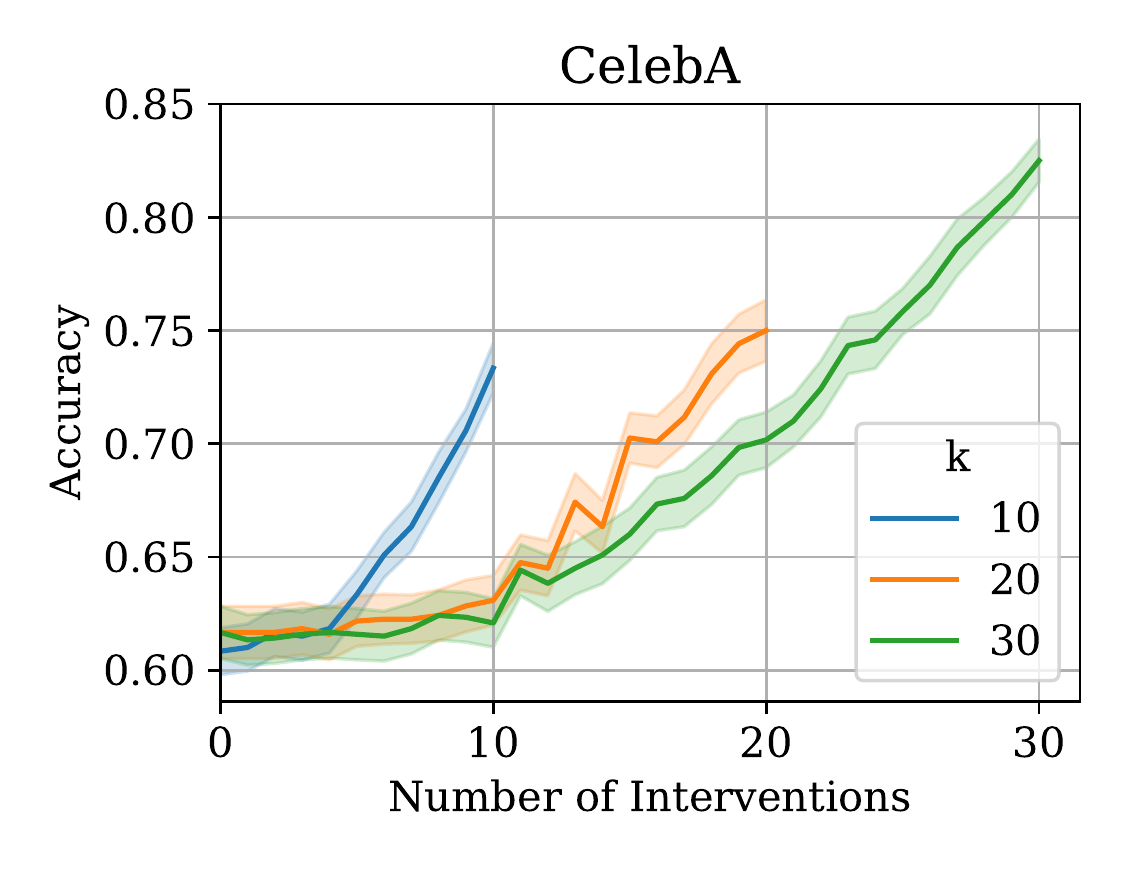}
    \caption{Prediction accuracy for CUB and CelebA vs number of oracle interventions, for certain values of $k$. Smaller concept sets have a greater increase in accuracy after the same number of interventions. Errors are $\pm 1 \sigma$, calculated over 10 random seeds which determines the concepts that are intervened.}
    \label{fig:intervene-acc}
\end{figure*}
\begin{figure*}[t!]
    \centering
    \includegraphics[width=\linewidth]{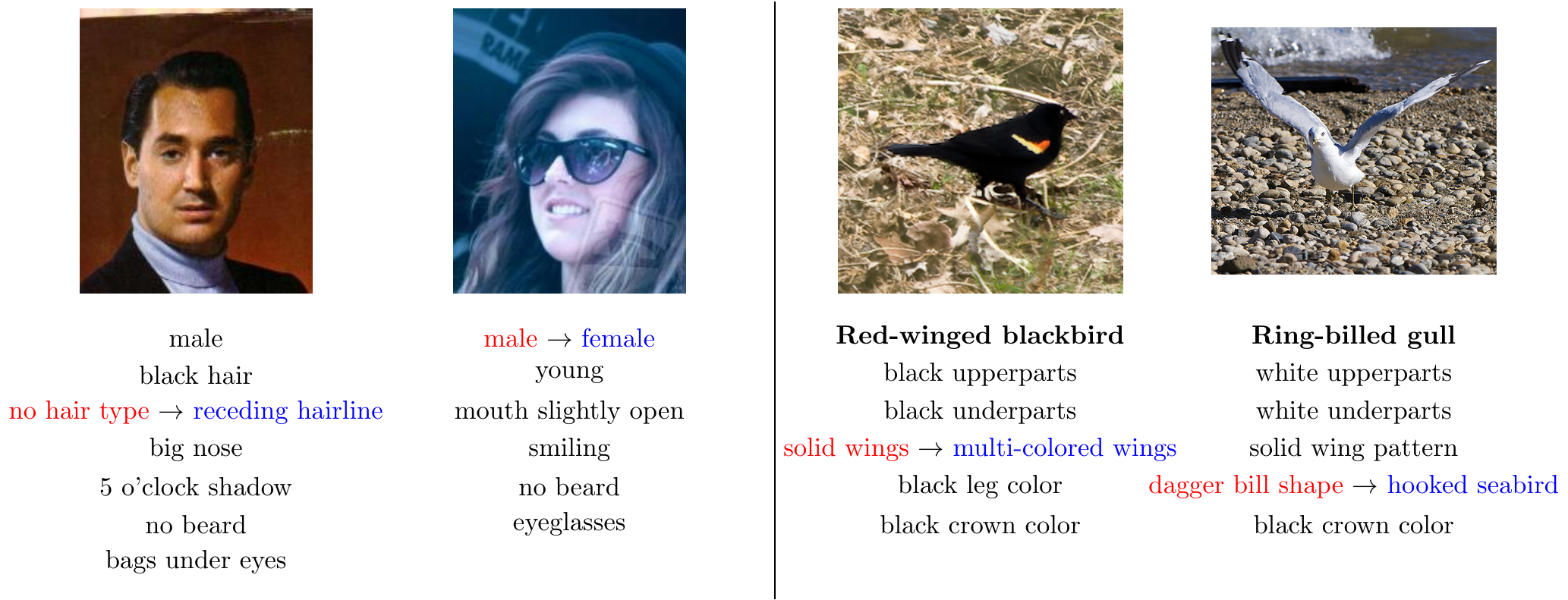}
    \caption{Example interventions using selected concept sets from CelebA (left, 8 concepts) and CUB (right, 6 concepts). The concept logit predictions are rounded to give binary predictions, and only the positive concepts are shown. Each intervention changes one concept which corrects the original incorrect model prediction. When using selected concept sets, it is much easier for experts to identify incorrect concepts and reduces the human cost of interventions.}
    \label{fig:intervene}
\end{figure*}

We first investigate how task accuracy compares between \scoms, which permit flexible test-time intervention \textit{without retraining} against the L2X baseline, which retrains a new model for each value of $k$. The results are shown in Figure \ref{fig:prediction-accuracy} and numerical values are provided in Tables \ref{tab:cub} and \ref{tab:celeba}.

For CUB and CelebA, the BE methods perform better than the greedier FS methods. On CUB, BE \scoms perform particularly well, achieving the optimal accuracy of 76.1\% when $k=6$. In addition, BE \scom outperforms the L2X baseline on CelebA for low numbers of concepts, with approximately equal performance for large concept sets.

\subsubsection{Intervention Accuracy}
One of the key benefits of \cbms is that they allow interventions from a human expert at the concept level to improve the accuracy of the model. To test this property for \scoms, concepts were replaced by predictions from an oracle, chosen randomly for simplicity. We assume that human experts are able to identify the true value of the concepts with perfect accuracy. The task of intervening on concepts is separate from the concept selection procedure performed in the human experiment.

The initial concept sets were selected using the BE variant of \scoms, since that gives the greatest prediction accuracy for both CUB and CelebA. For CUB, the oracle provides binary values for each concept, which are determined by the species that the bird belongs to. For CelebA, there is no clear class-level oracle, and individual image annotations are too noisy to be considered an oracle. Following work on human-derived soft annotations \citep{peterson2019human, Collins_Bhatt_Weller_2022}, we assume that the CelebA oracle intervenes on each concept with the ``correct uncertainty'', estimated as the mean value of that concept over all images of the same person. While we acknowledge that this oracle is unrealistic, we aim to show the efficacy of interventions assuming they can be obtained. Though beyond the scope of this work, we refer to ~\citeauthor{collins2023human} for an exposition of why such an oracle-intervention assumption may not always be wise in practice.

Figure \ref{fig:intervene-acc} highlights that the intervention accuracy increases more with each intervention for smaller concept set sizes ($k$). After only a few interventions, the accuracy gain for small $k$ overcomes the initial accuracy penalty compared to the full concept set. As each concept in a small concept set has a greater impact on the overall prediction, it is intuitive that each intervention yields a larger effect. Examples of interventions are shown in Figure \ref{fig:intervene}, where a single oracle intervention corrects the initial prediction made by the model. Since oracle interventions are more effective for smaller concept sets, \scoms exhibit greater intervenability compared to existing concept models. 

\subsection{Takeaways}
On the real-world datasets CUB and CelebA, \scoms perform better than existing methods. BE methods are much more accurate than FS, and require a fraction of the number of concepts for optimal performance.

\section{Human Experiment}
\label{sec:human-experiment}

\begin{figure}[t!]
    \centering
    \includegraphics[width=\linewidth]{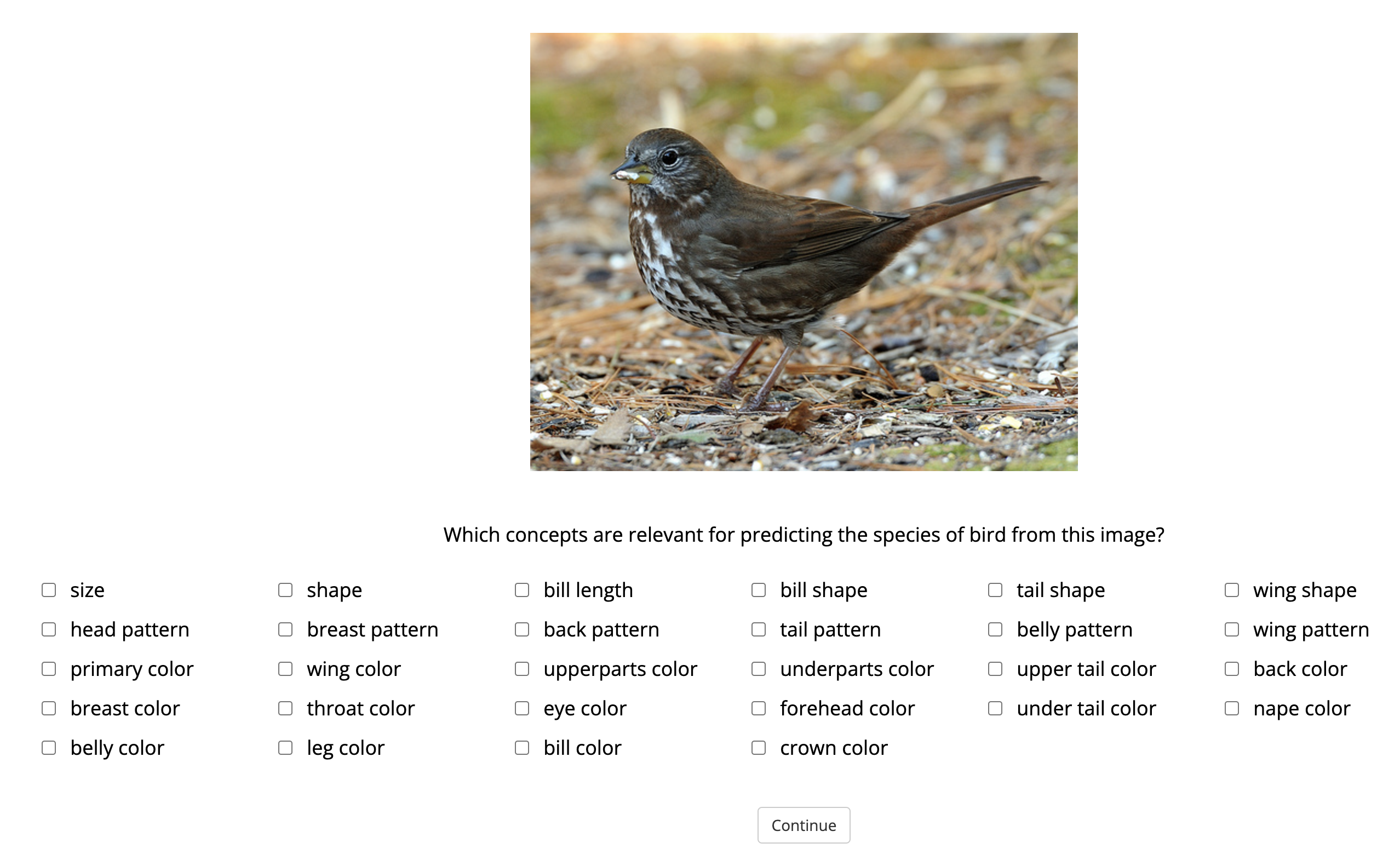}
    \caption{Interface for the human experiment. Participants were shown an image of a bird, and then asked to select the relevant concepts from the list of 28 total concepts.}
    \label{fig:interface}
\end{figure}
\begin{figure*}[t!]
    \centering
    \includegraphics[width=0.9\linewidth]{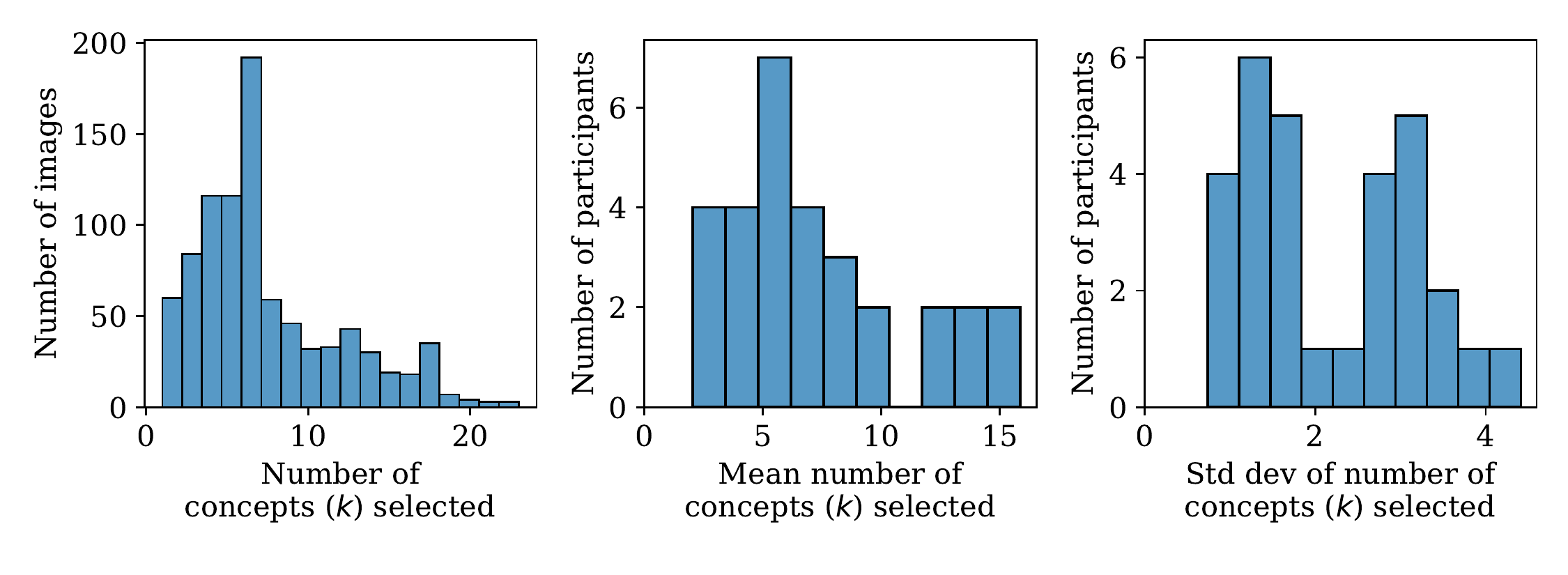}
    \caption{\textbf{Left:} Distribution of number of concepts selected for each image (left) and the mean number of concepts selected by each participant (center). Both distributions have a large variance, showing the number of concepts humans prefer to reason varies. The standard deviation of the number of concepts selected by each participant (right) is comparatively large, motivating the instance-level selection \scom provides.
}
    \label{fig:k-distribution}
\end{figure*}
\begin{figure*}[t!]
    \centering
    \includegraphics[width=\linewidth]{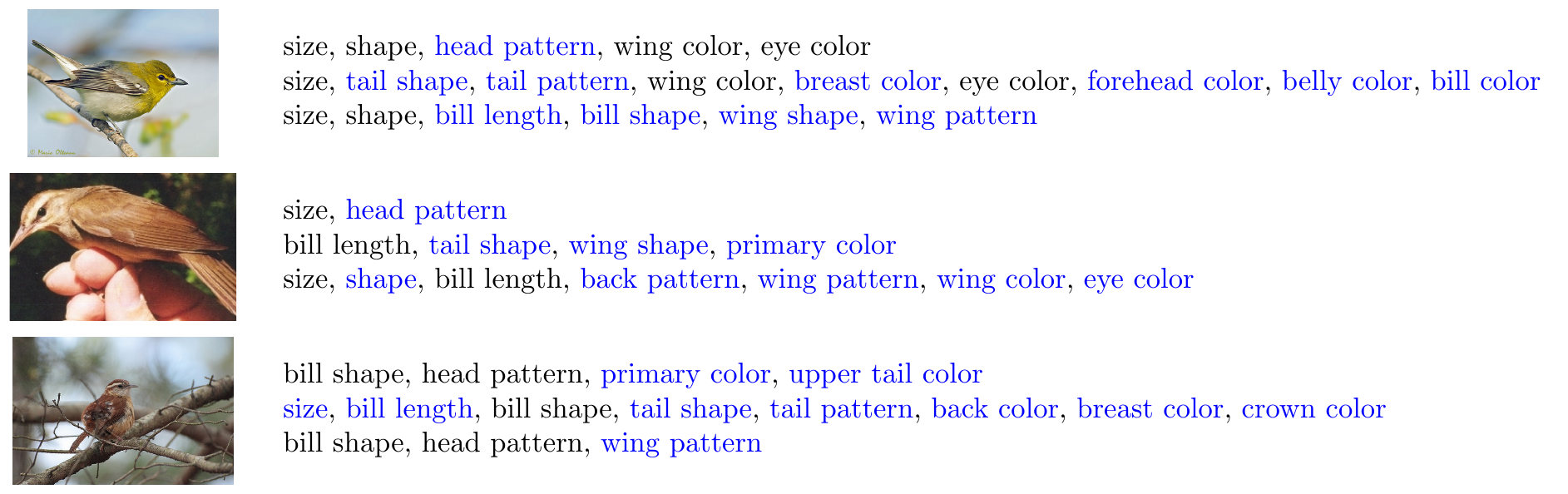}
    \caption{Example human selections of the most relevant concepts. Here, we depict cases where three participants labelled the same image. Concepts with at least one disagreement are shown in \textcolor{blue}{blue}, and illustrate that people have varied individual preferences.}
    \label{fig:duplicated}
\end{figure*}
\begin{figure*}[t!]
    \centering
    \includegraphics[width=0.9\linewidth]{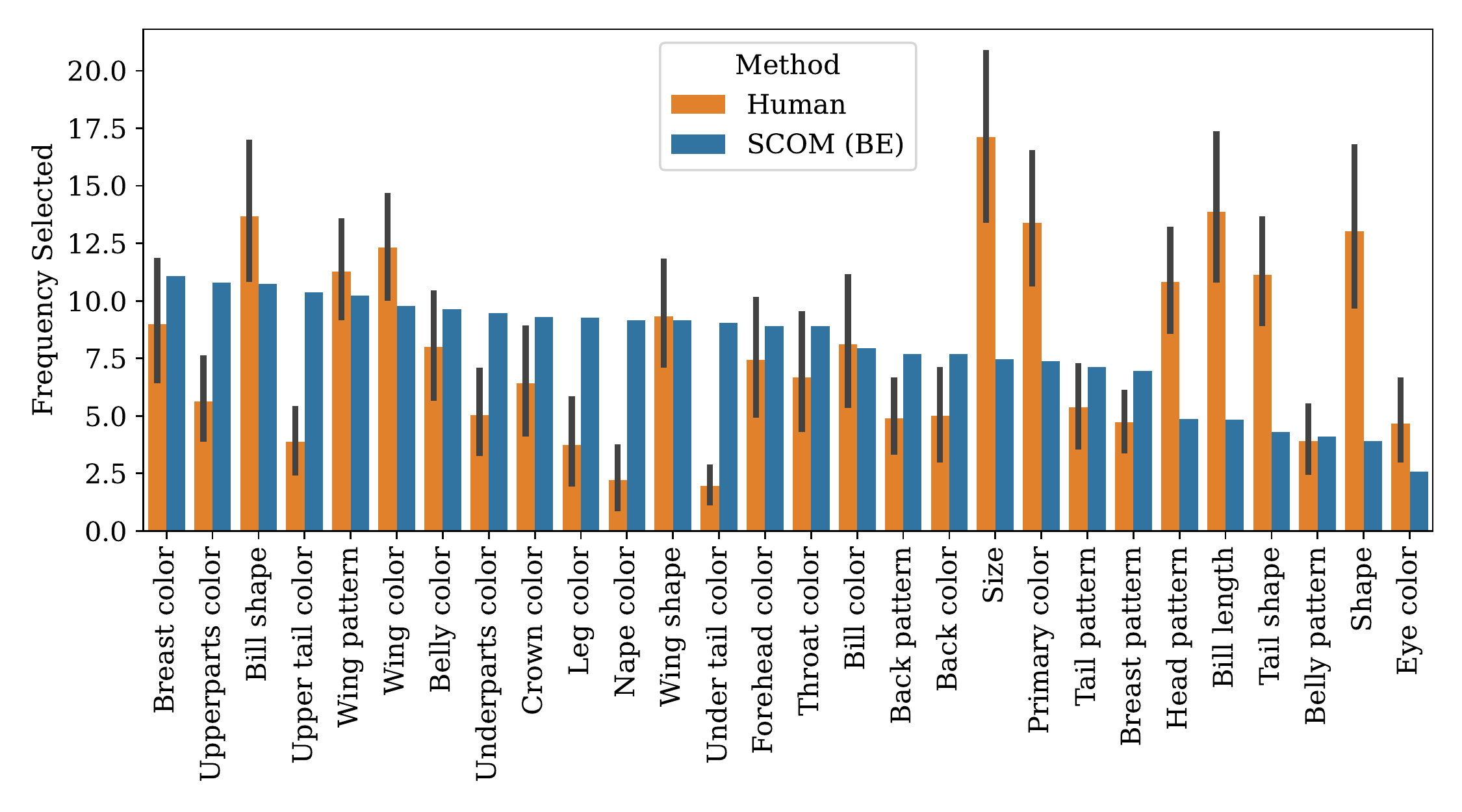}
    \caption{Bar chart comparing the frequency of concepts selected by participants to the selections by the BE \scom. Humans prefer concepts which \scoms rarely select, emphasising the importance of customisation for the concept sets selected by \scom. Errors bars of $\pm 1 \sigma$ are provided for the human selections.}
    \label{fig:team-selection}
\end{figure*}
Our computational results demonstrate the methodological prowess of \scoms. However, we care principally about \scoms to \textit{work with and support real humans}. One of the key motivations for \scoms is that they allow stakeholder customisation of concept sets \emph{at inference time}. The need for customisation rests on the assumption that humans have varied preferences for the number of concepts they prefer to reason with ($k$). Further, the desire for algorithmic selection of those concepts assumes that people may not always choose the most informative concepts on their own. We validate these two assumptions through a human subject experiment, wherein participants were asked to select the concept groups for each image. We release the resulting annotations in a new dataset, \cubsel.


\subsubsection{Experimental Details}
We recruited 30 US-based participants from the crowd-sourcing platform Prolific \citep{palan2018prolific}. Participants were asked to select as many concepts as they thought were relevant when predicting the bird species for each image. Selection implicitly captures their preference for the number of concepts to reason with ($k$). Each participant was shown 30 images, yielding a total of \textbf{900 concept set selections} (which compose \cubsel). The images were randomly sampled from the CUB test dataset, with a total of 28 possible concepts for each image. Participants were paid at a base rate of \$8/hour, with a bonus rate of up to \$9/hour to encourage high-quality selections. The study interface is shown in Figure \ref{fig:interface}. Further details on CUB are provided in the Computational Experiments.

\subsection{Results}
\subsubsection{Preferred Concept Set Sizes}
To motivate the customisation \scoms provide, we first investigate the distribution over $k$, i.e., the number of concepts people selected. Figure \ref{fig:k-distribution} shows the distribution of $k$ values both for each image and for each participant. A participant's $k$ was calculated as the mean $k$ selected over the 30 images they were shown. Both distributions have a large variance, showing that the number of concepts humans prefer to reason about varies significantly -- across people and individual images. There was a mean difference of 9.6 concepts when several participants labelled the same image, emphasising that different people prefer to reason with different concepts. Examples of multiple concept selections for the same image are shown in Figure \ref{fig:duplicated}. Thus, the customisation \scom provides is essential to cater for individual stakeholder preferences.

\subsubsection{Comparison Between Humans and \scoms}
\scoms assume that humans struggle to select the relevant concepts from a larger set. To evaluate the quality of human selections, we computed the accuracy of the \scom output model on CUB using human selections and define this as the ``human accuracy''. Importantly, although the concept sets are chosen by humans, the \scom model is still used for inference. Thus, the ``human accuracy'' is a measure of how informative concepts are \emph{for the model}; the information concepts provide may differ if people were responsible for classification instead.

We compare the human accuracy to \scom and L2X instance-level selection methods in Figure \ref{fig:human-acc}. Remarkably, the human accuracy is lower than random for all $k$ values and far lower than the BE \scom algorithms. To explore human versus model selection further, we compare the frequency of concepts selected by humans to the selections by a \scom\footnote{We use BE as the \scom baseline, evaluated on the same images shown to participants (which represent a subset of the test set).} in Figure \ref{fig:team-selection}. There is a significant difference between the distribution of concepts selected by humans vs \scom. Concepts such as ``Shape'' and ``Bill length'' are frequently selected by humans, but not favoured by \scom. Concepts which are not specific to bird species recognition are easier to understand and thus treated preferentially by humans, even if they are not the most informative concepts. Challenging concepts such as ``Nape color'' and ``Under tail color'' are rarely selected by humans, and some participants noted that they struggled to understand these concepts.

\subsection{Takeaways}
The human subject experiment shows that humans have individual $k$ values which vary significantly, motivating the test time customisation \scoms provide. Further, humans favour concepts which are less theoretically informative, as measured by \scom selections. Consequently, the accuracy of model predictions when using human selection is significantly lower than \scom selection.

\section{Discussion}
We investigated the performance of \scoms compared to random and L2X baselines, for two real world datasets. The selection algorithm uses MI to select the best concept set of a given size at either the dataset or instance level, and is robust to correlations between concepts. For CUB, prediction accuracy is far superior to random and human selections. For CelebA, \scoms outperform L2X, especially for low numbers of concepts. After a few interventions, predictions using the selected subset outperforms those using all concepts. Through our human subject experiment, we show that humans have different individual preferences for the number and type of concepts they prefer to reason about, motivating the value of \scoms. By avoiding retraining on different concept sets, \scoms allow stakeholders flexibility to customise the concepts used to make predictions \textit{after the model has been trained}. Next, we discuss implications of the results and potential limitations of \scoms.

\subsection{Choosing the Concept Set Size $k$ in Practice}
When \scoms are used by real stakeholders, there are several methods of choosing the concept set size $k$. A practitioner could employ a validation set and choose $k$ to either maximise the accuracy, or as the minimum value which surpasses a certain accuracy threshold. Alternatively, $k$ may be chosen according to the preferences of the experts who are intervening on the concepts.

\subsection{Human vs Algorithm Concept Preferences}
Our human experiment elucidated that humans and \scom selection algorithms disagree on which concepts should be selected; humans prefer generic concepts, while \scoms prefer specific concepts which provide more information. Although it decreases the human cost, allowing humans to customise the concept sets originally selected algorithmically is likely to \emph{decrease} the prediction accuracy as well. Hence, there is a compromise between accuracy and cost of intervention/interpretation which ought to be considered when using \scoms in practice.  Note, we acknowledge that the participants who contributed to \cubsel were not necessarily experts in the task. Experts may have different preferences when selecting concept sets, potentially changing the overall prediction accuracy as well.

\subsection{Limitations}
While our approach holds promise for improving the usability of concept models for stakeholders, it is worth noting conditions which permit maximal effectiveness: \scoms work best on datasets with a large number of concepts (e.g., CUB) and rely on a well-trained output model. In cases where all the concepts are required for good prediction, smaller concept sets selected by \scom will suffer an accuracy penalty. Our selection algorithms aim to minimise output entropy which is only a good measure of MI when the output model is \emph{well-trained}. Since the output model learns a mapping from every concept subset to the output, it may converge slower during training than models which only consider one concept set. Moreover, so far our work has only studied \scoms in the classification setting; it is worth studying \scoms in other paradigms (e.g., regression) and over additional real-world concept datasets.

When removing concepts at test-time, we recognise that while this removes explicit dependence from the model, it does not completely cleanse implicit dependence on the removed concepts. Sensitive concepts may be highly correlated with other concepts which are not removed, meaning some bias may still be present in the model. Further, the output model was trained on all the concepts, which each affect the parameters of the model even if they are later removed from selection. For concept sets initially selected, it is assumed that all concepts are available for selection. Removing a concept at test time invalidates this assumption and reduces the efficacy of the concept selection procedure.


\subsection{Future Work}
One of the benefits of our prediction method, which uses arbitrary concept subsets, is that it allows stakeholder customisation at inference-time without retraining. The human experiment in this work was focused on evaluating human concept selection. In reality, the concept sets are likely to be selected by \scom, and then provided to human experts for interpretation and intervention. An exciting direction following this work is to validate these usability benefits of \scoms through further human experiments. By reducing the cost of expert interventions, \scoms facilitate end-to-end experiments involving concept selection, stakeholder customisation and expert intervention. Such experiments would test the use of concept models in real-world settings, paving the way for the deployment of \scoms in industry.

We encourage further investigation into discrepancies between human and algorithm concept preferences -- and potential for complementarity~\citep{bayesianComp} -- through \cubsel and related datasets such as \texttt{CUB-S} \citep{collins2023human}.
In addition, smarter intervention methods, such as CooP \citep{chauhan2022interactive}, could be used in conjunction with \scoms to further increase the efficacy of interventions.

\section{Conclusion}
In this work, we have introduced \scoms, an algorithmic framework which makes predictions using a selected concept set for a given task. \scoms provide a way of improving the intervenability and interpretability of concept based models, with no compromise on accuracy. On the task of bird species classification, optimal accuracy (76.1\%) was achieved using only 6 out of the available 28 concepts. Oracle interventions are more effective for the selected concept sets, demonstrating greater intervenability for \scoms compared to existing concept models. 

Using a human subject experiment we have shown that the concept sets humans choose often differ from the concepts which are most useful to the model. \scoms account for this difference by allowing humans to customise the selected concept sets at inference time \emph{without retraining}, according to their individual preferences. \scoms place no restrictions on the form of concepts selected or the output model architecture, allowing for easy integration with other techniques. We release the annotations from our human experiment in a new dataset, \cubsel, to encourage further investigation into the utility of \scoms in concert with real humans. 

\section{Acknowledgments}

We thank the participants from Prolific who took part in our human experiment. KMC acknowledges support from the Marshall Commission and the Cambridge Trust. AW  acknowledges  support  from  a  Turing  AI  Fellowship  under grant  EP/V025279/1,  The  Alan  Turing  Institute,  and  the Leverhulme Trust via CFI. UB  acknowledges  support  from  DeepMind  and  the  Leverhulme Trust  via  the  Leverhulme  Centre  for  the  Future  of  Intelligence  (CFI),  and  from  the  Mozilla  Foundation.

\bibliography{main}

\begin{thebibliography}{41}
\providecommand{\natexlab}[1]{#1}

\bibitem[{Abid, Yuksekgonul, and Zou(2022)}]{abid2022meaningfully}
Abid, A.; Yuksekgonul, M.; and Zou, J. 2022.
\newblock Meaningfully debugging model mistakes using conceptual counterfactual
  explanations.
\newblock In \emph{International Conference on Machine Learning}, 66--88. PMLR.

\bibitem[{Alvarez~Melis and Jaakkola(2018)}]{alvarez2018towards}
Alvarez~Melis, D.; and Jaakkola, T. 2018.
\newblock Towards robust interpretability with self-explaining neural networks.
\newblock \emph{Advances in neural information processing systems}, 31.

\bibitem[{Bai et~al.(2022)Bai, Yeh, Ravikumar, Lin, and Hsieh}]{bai2022concept}
Bai, A.; Yeh, C.-K.; Ravikumar, P.; Lin, N.~Y.; and Hsieh, C.-J. 2022.
\newblock Concept Gradient: Concept-based Interpretation Without Linear
  Assumption.
\newblock \emph{arXiv preprint arXiv:2208.14966}.

\bibitem[{Chauhan et~al.(2022)Chauhan, Tiwari, Freyberg, Shenoy, and
  Dvijotham}]{chauhan2022interactive}
Chauhan, K.; Tiwari, R.; Freyberg, J.; Shenoy, P.; and Dvijotham, K. 2022.
\newblock Interactive Concept Bottleneck Models.
\newblock \emph{arXiv preprint arXiv:2212.07430}.

\bibitem[{Chen et~al.(2018)Chen, Song, Wainwright, and
  Jordan}]{chen2018learning}
Chen, J.; Song, L.; Wainwright, M.; and Jordan, M. 2018.
\newblock Learning to explain: An information-theoretic perspective on model
  interpretation.
\newblock In \emph{International Conference on Machine Learning}, 883--892.
  PMLR.

\bibitem[{Collins et~al.(2023)Collins, Barker, Zarlenga, Raman, Bhatt, Jamnik,
  Sucholutsky, Weller, and Dvijotham}]{collins2023human}
Collins, K.~M.; Barker, M.; Zarlenga, M.~E.; Raman, N.; Bhatt, U.; Jamnik, M.;
  Sucholutsky, I.; Weller, A.; and Dvijotham, K. 2023.
\newblock Human Uncertainty in Concept-Based AI Systems.
\newblock \emph{arXiv preprint arXiv:2303.12872}.

\bibitem[{Collins, Bhatt, and Weller(2022)}]{Collins_Bhatt_Weller_2022}
Collins, K.~M.; Bhatt, U.; and Weller, A. 2022.
\newblock Eliciting and Learning with Soft Labels from Every Annotator.
\newblock \emph{Proceedings of the AAAI Conference on Human Computation and
  Crowdsourcing}, 10(1): 40--52.

\bibitem[{Cowan(2001)}]{cowan2001magical}
Cowan, N. 2001.
\newblock The magical number 4 in short-term memory: A reconsideration of
  mental storage capacity.
\newblock \emph{Behavioral and brain sciences}, 24(1): 87--114.

\bibitem[{Crabb{\'e} and van~der Schaar(2022)}]{crabbe2022concept}
Crabb{\'e}, J.; and van~der Schaar, M. 2022.
\newblock Concept Activation Regions: A Generalized Framework For Concept-Based
  Explanations.
\newblock \emph{arXiv preprint arXiv:2209.11222}.

\bibitem[{Fleuret(2004)}]{fleuret2004fast}
Fleuret, F. 2004.
\newblock Fast binary feature selection with conditional mutual information.
\newblock \emph{Journal of Machine learning research}, 5(9).

\bibitem[{Ghorbani et~al.(2019)Ghorbani, Wexler, Zou, and
  Kim}]{ghorbani2019towards}
Ghorbani, A.; Wexler, J.; Zou, J.~Y.; and Kim, B. 2019.
\newblock Towards automatic concept-based explanations.
\newblock \emph{Advances in Neural Information Processing Systems}, 32.

\bibitem[{He et~al.(2016)He, Zhang, Ren, and Sun}]{he2016deep}
He, K.; Zhang, X.; Ren, S.; and Sun, J. 2016.
\newblock Deep residual learning for image recognition.
\newblock In \emph{Proceedings of the IEEE conference on computer vision and
  pattern recognition}, 770--778.

\bibitem[{Kim et~al.(2018)Kim, Wattenberg, Gilmer, Cai, Wexler, Viegas
  et~al.}]{kim2018interpretability}
Kim, B.; Wattenberg, M.; Gilmer, J.; Cai, C.; Wexler, J.; Viegas, F.; et~al.
  2018.
\newblock Interpretability beyond feature attribution: Quantitative testing
  with concept activation vectors (tcav).
\newblock In \emph{International conference on machine learning}, 2668--2677.
  PMLR.

\bibitem[{Kim et~al.(2022)Kim, Watkins, Russakovsky, Fong, and
  Monroy-Hern{\'a}ndez}]{kim2022help}
Kim, S.~S.; Watkins, E.~A.; Russakovsky, O.; Fong, R.; and
  Monroy-Hern{\'a}ndez, A. 2022.
\newblock " Help Me Help the AI": Understanding How Explainability Can Support
  Human-AI Interaction.
\newblock \emph{arXiv preprint arXiv:2210.03735}.

\bibitem[{Koh et~al.(2020)Koh, Nguyen, Tang, Mussmann, Pierson, Kim, and
  Liang}]{koh2020concept}
Koh, P.~W.; Nguyen, T.; Tang, Y.~S.; Mussmann, S.; Pierson, E.; Kim, B.; and
  Liang, P. 2020.
\newblock Concept bottleneck models.
\newblock In \emph{International Conference on Machine Learning}, 5338--5348.
  PMLR.

\bibitem[{Kolchinsky and Tracey(2017)}]{kolchinsky2017estimating}
Kolchinsky, A.; and Tracey, B.~D. 2017.
\newblock Estimating mixture entropy with pairwise distances.
\newblock \emph{Entropy}, 19(7): 361.

\bibitem[{Kotsiantis, Kanellopoulos, and Pintelas(2006)}]{kotsiantis2006data}
Kotsiantis, S.; Kanellopoulos, D.; and Pintelas, P. 2006.
\newblock Data Preprocessing for Supervised Learning.
\newblock \emph{International Journal of Computer Science}, 1: 111--117.

\bibitem[{Kraskov, St{\"o}gbauer, and
  Grassberger(2004)}]{kraskov2004estimating}
Kraskov, A.; St{\"o}gbauer, H.; and Grassberger, P. 2004.
\newblock Estimating mutual information.
\newblock \emph{Physical review E}, 69(6): 066138.

\bibitem[{Lewis(1992)}]{lewis1992representation}
Lewis, D.~D. 1992.
\newblock \emph{Representation and learning in information retrieval}.
\newblock University of Massachusetts Amherst.

\bibitem[{Liu et~al.(2015)Liu, Luo, Wang, and Tang}]{liu2015faceattributes}
Liu, Z.; Luo, P.; Wang, X.; and Tang, X. 2015.
\newblock Deep Learning Face Attributes in the Wild.
\newblock In \emph{Proceedings of International Conference on Computer Vision
  (ICCV)}.

\bibitem[{Luck and Vogel(1997)}]{luck1997capacity}
Luck, S.~J.; and Vogel, E.~K. 1997.
\newblock The capacity of visual working memory for features and conjunctions.
\newblock \emph{Nature}, 390(6657): 279--281.

\bibitem[{Marconato, Passerini, and Teso(2022)}]{marconato2022glancenets}
Marconato, E.; Passerini, A.; and Teso, S. 2022.
\newblock GlanceNets: Interpretabile, Leak-proof Concept-based Models.
\newblock \emph{arXiv preprint arXiv:2205.15612}.

\bibitem[{Miller(1956)}]{miller1956magical}
Miller, G.~A. 1956.
\newblock The magical number seven, plus or minus two: Some limits on our
  capacity for processing information.
\newblock \emph{Psychological review}, 63(2): 81.

\bibitem[{Nemhauser, Wolsey, and Fisher(1978)}]{nemhauser1978analysis}
Nemhauser, G.~L.; Wolsey, L.~A.; and Fisher, M.~L. 1978.
\newblock An analysis of approximations for maximizing submodular set
  functions—I.
\newblock \emph{Mathematical programming}, 14(1): 265--294.

\bibitem[{Nixon and Aguado(2019)}]{nixon2019feature}
Nixon, M.; and Aguado, A. 2019.
\newblock \emph{Feature extraction and image processing for computer vision}.
\newblock Academic press.

\bibitem[{Palan and Schitter(2018)}]{palan2018prolific}
Palan, S.; and Schitter, C. 2018.
\newblock Prolific. ac—A subject pool for online experiments.
\newblock \emph{Journal of Behavioral and Experimental Finance}, 17: 22--27.

\bibitem[{Perez and Wang(2017)}]{perez2017effectiveness}
Perez, L.; and Wang, J. 2017.
\newblock The effectiveness of data augmentation in image classification using
  deep learning.
\newblock \emph{arXiv preprint arXiv:1712.04621}.

\bibitem[{Peterson et~al.(2019)Peterson, Battleday, Griffiths, and
  Russakovsky}]{peterson2019human}
Peterson, J.~C.; Battleday, R.~M.; Griffiths, T.~L.; and Russakovsky, O. 2019.
\newblock Human uncertainty makes classification more robust.
\newblock In \emph{Proceedings of the IEEE/CVF International Conference on
  Computer Vision}, 9617--9626.

\bibitem[{Ramaswamy et~al.(2022)Ramaswamy, Kim, Fong, and
  Russakovsky}]{ramaswamy2022overlooked}
Ramaswamy, V.~V.; Kim, S.~S.; Fong, R.; and Russakovsky, O. 2022.
\newblock Overlooked factors in concept-based explanations: Dataset choice,
  concept salience, and human capability.
\newblock \emph{arXiv preprint arXiv:2207.09615}.

\bibitem[{Shin et~al.(2022)Shin, Jo, Ahn, and Lee}]{shin2022closer}
Shin, S.; Jo, Y.; Ahn, S.; and Lee, N. 2022.
\newblock A Closer Look at the Intervention Procedure of Concept Bottleneck
  Models.
\newblock In \emph{Workshop on Trustworthy and Socially Responsible Machine
  Learning, NeurIPS 2022}.

\bibitem[{Steyvers et~al.(2022)Steyvers, Tejeda, Kerrigan, and
  Smyth}]{bayesianComp}
Steyvers, M.; Tejeda, H.; Kerrigan, G.; and Smyth, P. 2022.
\newblock Bayesian modeling of human–AI complementarity.
\newblock \emph{Proceedings of the National Academy of Sciences}, 119(11):
  e2111547119.

\bibitem[{Tenenbaum(1998)}]{tenenbaum1998bayesian}
Tenenbaum, J. 1998.
\newblock Bayesian modeling of human concept learning.
\newblock \emph{Advances in neural information processing systems}, 11.

\bibitem[{Wah et~al.(2011)Wah, Branson, Welinder, Perona, and
  Belongie}]{wah2011caltech}
Wah, C.; Branson, S.; Welinder, P.; Perona, P.; and Belongie, S. 2011.
\newblock The caltech-ucsd birds-200-2011 dataset.

\bibitem[{Welinder et~al.(2010)Welinder, Branson, Mita, Wah, Schroff, Belongie,
  and Perona}]{welinder2010caltech}
Welinder, P.; Branson, S.; Mita, T.; Wah, C.; Schroff, F.; Belongie, S.; and
  Perona, P. 2010.
\newblock Caltech-UCSD birds 200.

\bibitem[{Yeh et~al.(2020)Yeh, Kim, Arik, Li, Pfister, and
  Ravikumar}]{yeh2020completeness}
Yeh, C.-K.; Kim, B.; Arik, S.; Li, C.-L.; Pfister, T.; and Ravikumar, P. 2020.
\newblock On completeness-aware concept-based explanations in deep neural
  networks.
\newblock \emph{Advances in Neural Information Processing Systems}, 33:
  20554--20565.

\bibitem[{Yoon, Jordon, and van~der Schaar(2019)}]{yoon2019invase}
Yoon, J.; Jordon, J.; and van~der Schaar, M. 2019.
\newblock INVASE: Instance-wise variable selection using neural networks.
\newblock In \emph{International Conference on Learning Representations}.

\bibitem[{Yuksekgonul, Wang, and Zou(2022)}]{yuksekgonul2022post}
Yuksekgonul, M.; Wang, M.; and Zou, J. 2022.
\newblock Post-hoc Concept Bottleneck Models.
\newblock \emph{arXiv preprint arXiv:2205.15480}.

\bibitem[{Zarlenga et~al.(2022)Zarlenga, Barbiero, Ciravegna, Marra, Giannini,
  Diligenti, Shams, Precioso, Melacci, Weller et~al.}]{zarlenga2022concept}
Zarlenga, M.~E.; Barbiero, P.; Ciravegna, G.; Marra, G.; Giannini, F.;
  Diligenti, M.; Shams, Z.; Precioso, F.; Melacci, S.; Weller, A.; et~al. 2022.
\newblock Concept embedding models.
\newblock \emph{arXiv preprint arXiv:2209.09056}.

\bibitem[{Zhang and Zhou(2013)}]{zhang2013review}
Zhang, M.-L.; and Zhou, Z.-H. 2013.
\newblock A review on multi-label learning algorithms.
\newblock \emph{IEEE transactions on knowledge and data engineering}, 26(8):
  1819--1837.

\bibitem[{Zheng and Casari(2018)}]{zheng2018feature}
Zheng, A.; and Casari, A. 2018.
\newblock \emph{Feature engineering for machine learning: principles and
  techniques for data scientists}.
\newblock " O'Reilly Media, Inc.".

\bibitem[{Zhou et~al.(2018)Zhou, Sun, Bau, and
  Torralba}]{zhou2018interpretable}
Zhou, B.; Sun, Y.; Bau, D.; and Torralba, A. 2018.
\newblock Interpretable basis decomposition for visual explanation.
\newblock In \emph{Proceedings of the European Conference on Computer Vision
  (ECCV)}, 119--134.

\end{thebibliography}

\appendix
\section{Theorems}
\begin{theorem}
\label{theorem:output-entropy}
For a well-trained estimator, minimising the output entropy conditioned on the concepts, $H(\hat{Y}|C)$, is equivalent to maximising the MI, $I(Y; C)$:
\begin{equation}
    \argmincz H(\hat{Y}|C) = \argmaxcz I(Y; C)
\end{equation}
\end{theorem}
\begin{proof}
Using the independence of $H(Y)$ with respect to $C$:
\begin{equation}
    \argmincz H(\hat{Y}|C) = \argmincz H(\hat{Y}|C) -  H(Y)
\end{equation}
For a well-trained estimator, $p(y|c) = p(\hat{y}|c)$, implying $H(Y|C) = H(\hat{Y}|C)$. Using this assumption:
\begin{equation}
    \argmincz H(\hat{Y}|C) = \argmincz H(Y|C) -  H(Y)
\end{equation}
Noting that the RHS is the negative MI:
\begin{equation}
    \argmincz H(\hat{Y}|C) = \argmincz - I(Y; C)
\end{equation}
Reversing the minus sign gives the required result:
\begin{equation}
    \argmincz H(\hat{Y}|C) = \argmaxcz I(Y; C)
\end{equation}
\end{proof}

\section{Further Results}
\begin{figure}[h!]
    \centering
    \includegraphics[width=\linewidth]{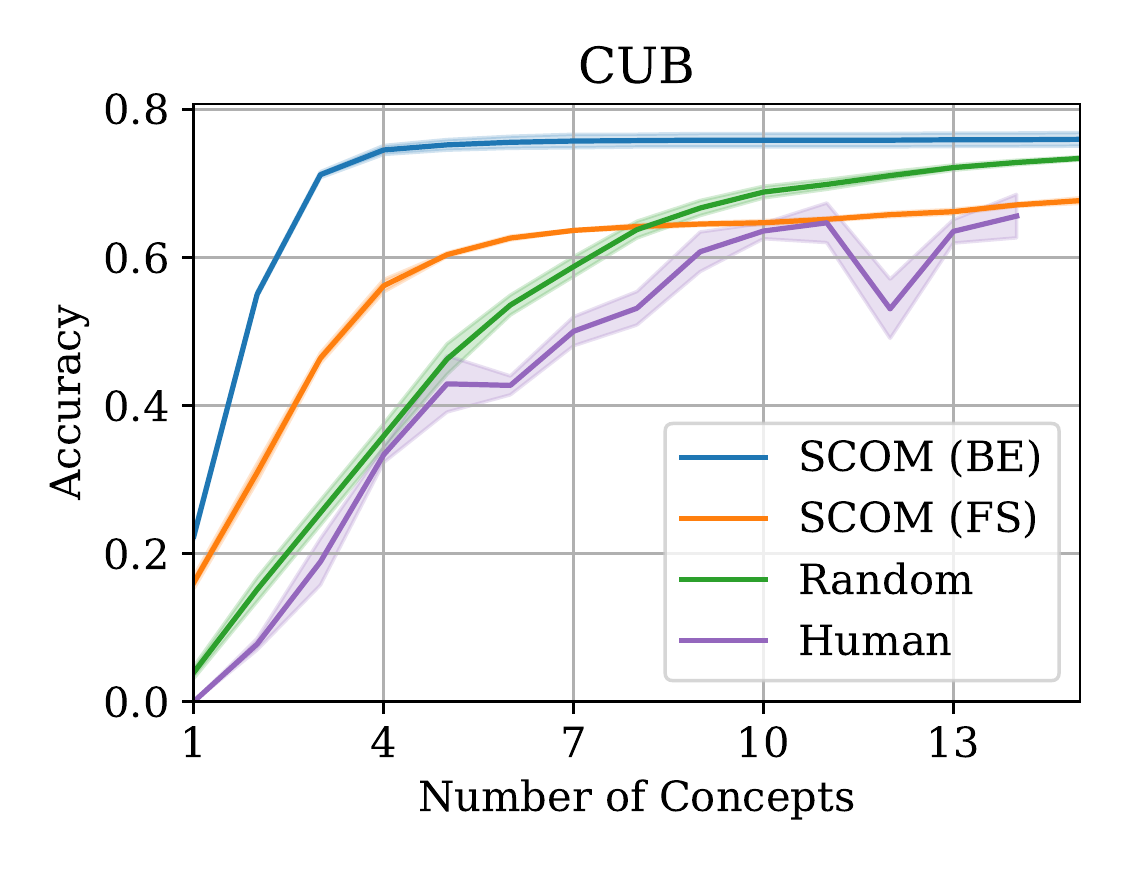}
    \caption{Prediction accuracy on the CUB dataset as a function of $k$. Accuracy using human selections is even worse than random, motivating the need for \scoms to provide algorithmic selections. Errors are $\pm 1\sigma$, calculated over random seeds which affects the concepts randomly selected and the concept model chosen.}
    \label{fig:human-acc}
\end{figure}

\begin{table}[h!]
\centering
\begin{tabular}{c|cccc}
\toprule
CUB & \multicolumn{4}{c}{Method} \\
$k$   & BE             & FS             & Human          & Random         \\ \midrule
1        & $22.4 \pm 1$   & $16.1 \pm 1$   & $0.0 \pm 0$    & $4.4 \pm 4$    \\
3        & $71.1 \pm 0.7$ & $46.4 \pm 1$   & $18.9 \pm 6$   & $26.9 \pm 7$   \\
 6        & $75.5 \pm 1$   & $62.6 \pm 1$   & $42.7 \pm 2$   & $55.5 \pm 6$   \\
12       & $75.8 \pm 2$   & $65.7 \pm 0.5$ & $53.0 \pm 0.7$ & $71.3 \pm 2$   \\
28 (All) & $75.3 \pm 0.6$ & $75.3 \pm 0.6$ & -              & $75.3 \pm 0.6$ \\ \bottomrule
\end{tabular}
\caption{Prediction accuracy on the CUB dataset for \scoms compared to the random and human benchmarks for selected values of $k$. Errors are $\pm 1 \sigma$, calculated over 3 random seeds which determine the concept logit predictions and the initial parameters of the output model.}
\label{tab:cub}
\end{table}

\begin{table}[h!]
\centering
\begin{tabular}{c|ccccc}
\toprule
CelebA & \multicolumn{4}{c}{Method} \\
$k$ & BE            & FS            & L2X          & Random       \\ \midrule
1        & $35.8 \pm 10$ & $30.0 \pm 10$ & $15.9 \pm 6$ & $8.8 \pm 2$  \\
5        & $60.8 \pm 8$  & $42.5 \pm 8$  & $45.8 \pm 6$ & $22.0 \pm 5$ \\
10       & $63.3 \pm 2$  & $45.8 \pm 9$  & $56.2 \pm 4$ & $35.5 \pm 5$ \\
20       & $65.0 \pm 2$  & $48.3 \pm 10$ & $62.9 \pm 2$ & $51.2 \pm 5$ \\
40 (All) & $63.8 \pm 3$  & $63.8 \pm 3$  & $63.8 \pm 3$ & $63.8 \pm 3$ \\ \bottomrule
\end{tabular}
\caption{Prediction accuracy on the CelebA dataset for \scoms compared to the random and L2X instance benchmarks for selected values of $k$. Errors are $\pm 1 \sigma$, calculated over 5 random seeds which determine the train/test split and the initial parameters of the output model.}
\label{tab:celeba}
\end{table}

\end{document}